\documentclass[11pt,reqno]{amsart}
\usepackage[utf8]{inputenc}
\usepackage{setspace}

\usepackage{amsmath,amssymb,amsthm,mathrsfs,graphicx,caption,booktabs,multirow,adjustbox,makecell}
\usepackage[margin=3.4cm]{geometry}
\usepackage{amsthm}

\newtheoremstyle{tight}
{0pt}   
{0pt}   
{}      
{}      
{\bfseries} 
{.}     
{0.5em} 
{}      

\theoremstyle{tight}

\newtheorem{theorem}{Theorem}[section]
\newtheorem{example}{Example}[section]
\newtheorem{definition}{Definition}[section]
\newtheorem{corollary}{Corollary}[section]
\newtheorem{property}{Property}[section]
\title[Cumulative Residual Mathai--Haubold Entropy]{The Cumulative Residual Mathai--Haubold Entropy and its Non-parametric Inference}

\author[]{A\lowercase{nija} C.R.\lowercase{\textsuperscript{a}} , S\lowercase{mitha} S.\lowercase{\textsuperscript{a}} \lowercase{and} S\lowercase{udheesh} K. K\lowercase{attumannil.\textsuperscript{b}}
\\
\lowercase{\textsuperscript{a}}K\lowercase{uriakose} E\lowercase{lias} C\lowercase{ollege},
  M\lowercase{annanam},
  K\lowercase{erala},
  I\lowercase{ndia,}\\
\lowercase{\textsuperscript{b}}I\lowercase{ndian} S\lowercase{tatistical} I\lowercase{nstitute},
  C\lowercase{hennai}, I\lowercase{ndia.}}
\begin{document}
\maketitle
\doublespacing
\vspace{-0.2in}
\begin{abstract}
   We introduce the cumulative residual Mathai--Haubold entropy (CRMHE) and investigate its properties. We then propose a dynamic counterpart, the dynamic cumulative residual Mathai--Haubold entropy (DCRMHE), and establish its uniqueness in characterizing the distribution function. Non-parametric estimators for the CRMHE and DCRMHE are developed based on the kernel density estimation of the survival function. The efficacy of the estimators is assessed through a comprehensive Monte Carlo simulation study. The relevance of the proposed DCRMHE estimator is illustrated using two real-world datasets: on the failure times of 70 aircraft windshields and failure times of 40 randomly selected mechanical switches.
\end{abstract}
\section{Introduction}
Measuring the uncertainty of a random variable is a fundamental problem in information theory. Shannon (1948) introduced the concept of entropy, laying the foundation for information theory. It measures the average uncertainty or information content associated with the outcomes of a random variable.

 Let $X$ be a non-negative random variable having probability density function $f(x)$, then the Shannon entropy is defined as
 \begin{align}\label{entropy}
 H(X)=-\int_{0}^{\infty} f(x) \log(f(x)) \, dx  = E[-\log f(X)],
\end{align}
 where “log” indicates the natural logarithm. It has been widely used in various fields, such as communication theory, computer science and the physical and biological sciences.
\par By introducing additional parameters, various generalizations of Shannon entropy have been proposed in the literature, making these measures more sensitive to different shapes of probability distributions.
 Mathai and Haubold (2006) introduced a generalized information measure, namely Mathai--Haubold entropy (MHE) as
 \begin{align}\label{MHE}
M_\alpha(X)= \frac{1}{\alpha-1}\left( {\int_0^\infty} f^{2-\alpha}(x) \,dx -1\right), \alpha \neq 1, 0<\alpha<2.
 \end{align}
The generalizing parameter $\alpha$ provides greater flexibility in quantifying uncertainty, making the MHE a more general measure than Shannon entropy.
 Further, when $\alpha \to 1$, $M_\alpha(X)$  reduces to the Shannon entropy given in (\ref{entropy}).

 \par To measure the uncertainty associated with the remaining lifetime of a unit, the concept of residual entropy was introduced by Ebrahimi (1996) and is defined as
\begin{align}\label{residual entropy}
H(X;t) =-\int_t^\infty \frac{f(x)}{\bar{F}(t)}\log \left( \frac{f(x)}{\bar{F}(t)}\right) \,dx,
\end{align}
where $\bar{F}(t)=1-F(t)$ is the survival function.
This measure is used for modeling and analysis of lifetime data.
\par In a similar manner, Dar and Al--Zahrani (2013) extended the Mathai--Haubold entropy measure to account for the current age of the system, into consideration as Mathai--Haubold residual entropy, which is given as
 \begin{align}
 M_\alpha(X;t)= \frac{1}{\alpha-1}\left( {\int_t^\infty} \left(\frac{{f}(x)}{\bar{F}(t)}\right)^{2-\alpha} \,dx -1\right), \alpha \neq 1, 0<\alpha<2,t>0.
 \end{align}
\par A number of information measures initially formulated via probability density functions have subsequently been reconstructed using the survival function. This methodological shift is driven by the survival function's superior utility and interpretability in reliability and survival analysis, particularly when handling censored or truncated data. Within this framework, Rao et al. (2004) introduced a seminal alternative entropy measure known as cumulative residual entropy (CRE), which is given by
\begin{align}\label{CRE}
     CRE(X)=-\int_{0}^{\infty} \bar{F}(x) \log \bar{F}(x) \,dx.
\end{align}

  The CRE is always non-negative and has practical applications in fields such as reliability analysis and image alignment.
 \par Asadi and Zohrevand (2007) proposed a modification of the cumulative residual entropy (CRE), referred to as the dynamic cumulative residual entropy and is defined as
 \begin{align}\label{CREt}
 CRE(X;t)=-\int_{t}^{\infty} \frac{\bar{F}(x)} {\bar{F}(t)} \log \left( \frac{\bar{F}(x)}{\bar{F}(t)} \right) \,dx.
 \end{align}
 \par Rajesh and Sunoj (2016) proposed an alternative form of cumulative Tsallis entropy of order $\alpha$ as an extension of Tsallis entropy (Tsallis, 1988) and they also defined its dynamic version as
 \begin{align}\label{zeta}
 \zeta_{\alpha}(t)=\zeta_{\alpha}(X;t)=\frac{1}{\alpha-1}\left(r(t)-\int_{t}^{\infty}\left(\frac{\bar{F}(x)}{\bar{F}(t)}\right)\,dx\right),
 \end{align}
 where $r(t)$ denotes the mean residual life function of the random variable $X$ and $\zeta_{\alpha}(t)$ measures the uncertainty for the residual random variable  in terms of the survival function for different $\alpha$.
 \par Sudheesh et al. (2022) introduced a generalized measure of cumulative residual entropy and studied its properties. They also showed that several existing measures of entropy are special cases of this generalized cumulative entropy.
  \par Based on the entropy generating function introduced by Golomb (1966), Smitha et al. (2023) proposed an alternative definition known as the cumulative residual entropy generating function (CREGF), which is derived from the survival function. The authors also extended this concept to create a dynamic counterpart, termed the dynamic cumulative residual entropy generating function (DCREGF)
  \begin{align}\label{C(s)}
  C_s{(X;t)}=\int_t^\infty\left(\frac{\bar{F}(x)}{\bar{F}(t)}\right)^s \,dx, s\geq0,s\neq1.
 \end{align}


Sudheesh et al. (2023) established some relationships between the information measures and the Gini mean difference(GMD). They also presented relationships between the dynamic versions of the cumulative residual/past extropy measures and the truncated GMD. Recently, Lu and Xie (2025)  introduced the concept of the weighted cumulative residual information generating function (WCRIGF) and investigate its key properties, including various bounds, its relationships with other information measures, and its behavior under monotonic transformations. Additionally, they introduced the relative weighted cumulative residual information generating function (RWCRIGF) and systematically examined its properties. Pandey  and Kundu (2025) generalized the concept of residual extropy (dual concept of entropy) to the bivariate case by introducing a new information-theoretic measure that captures the uncertainty of a paired lifetime beyond given thresholds. Balakrishnan et al. (2024) proposed dispersion indices based on the Kerridge's inaccuracy measure and the Kullback–Leibler divergence measure. Building on this idea, Sankaran et al. (2025) introduced a quantile-based dispersion index derived from the Kullback–Leibler divergence, offering a useful alternative for examining variability in information measures.

\par Motivated by the superior stability of cumulative distribution based information measures over their density based counterparts, this work introduces an extension of the Mathai--Haubold entropy formulated using the survival function. As a result of this stability, the proposed measure is well-suited for applications in reliability assessment, maintenance scheduling, and the modeling of lifetime data.


The remaining sections of the paper are organized as follows. Section 2 introduces the cumulative residual Mathai--Haubold entropy (CRMHE) and explores its properties. Section 3 presents the dynamic cumulative residual Mathai--Haubold entropy (DCRMHE), discusses its properties, and demonstrates that it uniquely characterizes the distribution. In Section 4, the characterization results are examined by establishing connections between the hazard rate, mean residual life function, and DCRMHE. Section 5 introduces new classes of lifetime distributions and, based on DCRMHE, proposes a hazard rate ordering. Section 6 focuses on the non-parametric kernel estimation methods for CRMHE and DCRMHE while Section 7 evaluates the performance of these estimators through Monte Carlo simulation studies. In Section 8, two real-data applications involving the failure times of aircraft windshields and randomly selected mechanical switches are presented. Finally, Section 9 offers a summary and perspectives for future research.
\section{Cumulative residual  Mathai--Haubold  entropy (CRMHE)}

 In this section, we introduce the cumulative residual  Mathai--Haubold  entropy as a survival based extension of  Mathai--Haubold  entropy and then examine some of its properties.
 \begin{definition}
 For a continuous non-negative random variable $X$ with survival function $\bar{F}(x)$, CRMHE denoted by $CRM_\alpha(X)$ is defined as
 \begin{align}\label{CRM}
CRM_\alpha(X)= \frac{1}{\alpha-1}\left[ {\int_0^\infty} (\bar{F}(x))^{2-\alpha} \,dx -1 \right], \alpha \neq 1, 0<\alpha<2.
 \end{align}
\end{definition}
 As $\alpha \to 1$, $CRM_\alpha(X)$ reduces to the cumulative residual entropy (CRE) given in (\ref{CRE}).

\noindent Next, we establish some properties of $CRM_\alpha(X)$. The following property shows that the measure remains invariant under relocation.    \begin{property}
Let the random variable $Y=aX+b$, where $a>0$ and $b\geq0$. Then, the cumulative residual  Mathai--Haubold  entropy is
 \begin{align}\label{CRM Y}
CRM_\alpha(Y)=aCRM_\alpha(X)+\frac{a-1}{\alpha-1}.
\end{align}
\end{property}
\begin{proof}
 We have
 \begin{align}\label{cr}
CRM_\alpha(Y)=  \frac{1}{\alpha-1}\left( {\int_0^\infty} (\bar{F}_Y{(y)})^{2-\alpha} \,dy -1\right).
 \end{align}
 Here $Y=aX+b$, 
 so
\begin{align*}
CRM_\alpha(Y)=\frac{1}{\alpha-1}\left(aCRM_\alpha(X)(\alpha-1)+a-1\right).
 \end{align*}
On simplification, we get the required result.
 \end{proof}
\begin{corollary}
From the above property, it is clear that, if $a=1$
\[
CRM_\alpha(Y)=CRM_\alpha(X).
\]
\end{corollary}
The next property gives a bound for the cumulative residual  Mathai--Haubold  entropy (CRMHE) with reference to the mean as follows.
\begin{property}
Consider the non-negative random variable $X$ with mean $\mu$, then  \\
$(i)$ $CRM_\alpha(X)>\frac{\mu-1}{\alpha-1}$, for $0<\alpha<1$.\\
$(ii)$ $CRM_\alpha(X)<\frac{\mu-1}{\alpha-1}$, for $1<\alpha<2$.
\end{property}
\begin{proof}
  For $0<\alpha<1$, we have
  \[
 (\bar{F}(x))^{2-\alpha} < \bar{F}(x).
  \]
  Integrate with respect to $x$, we get
  \[
  {\int_0^\infty}((\bar{F}(x))^{2-\alpha}\,dx<{\int_0^\infty}\bar{F}(x)\,dx.
  \]
From the above inequality, we can observe that
 \[
  CRM_\alpha(X) 
  >\frac{\mu-1}{\alpha-1},
  \]
  where
  \[
 \mu={\int_0^\infty} \bar{F}(x)\,dx.
 \]
 This completes $(i)$. Similarly for $(ii)$.
  \end{proof}
 Under the proportional hazards (PH) model assumption, the survival functions of the random variables $X$ and $X_{\theta}^*$ satisfy the following relationship
  $$\bar{F}_\theta^*(x)=(\bar{F}(x))^\theta, \theta>0,x \in R.$$
The following property establishes a relationship between the CRMHE of the random variables $X$ and $X_{\theta}^*$.
\begin{property}
$CRM_\alpha(X_\theta^*)=\frac{\beta-1}{\alpha-1}CRM_\beta(X)$,
where $\beta=2-\theta(2-\alpha)$.
\end{property}
\begin{example}
Consider a non-negative continuous random variable $X$ with distribution function $F$ and let $X_{1:n}$ be the first order statistic based on the random sample $X_1,X_2,...,X_n$ from $F$. So we have $\bar{F}_{X_{1:n}}(x)=(\bar{F}(x))^n$. Observe that $X_i,i=1,2,...,n$ and  $X_{1:n}$ satisfies the condition for being a proportional hazards model. Further in view of cumulative residual  Mathai--Haubold entropy
\[
CRM_\alpha(X_{1:n})=\frac{\beta-1}{\alpha-1}CRM_\beta(X), \text{ where }  \beta=2-n(2-\alpha),
\]
we get the  relation between the CRMHE of two
random variables as claimed in property 2.3.
\end{example}
\par The implication of the above result is that when a system of components with life distribution $F(x)$ are in series, the cumulative residual  Mathai--Haubold  entropy of minimum depends on the number of components and the cumulative residual  Mathai--Haubold  entropy of the original distribution $F(x).$
\begin{example}
If
$X$ follows exponential distribution with parameter $\lambda$ then
\[
CRM_\alpha(X_\theta^*)
=\frac{\beta-1}{\alpha-1}CRM_\beta(X);\beta=2-\theta(2-\alpha).
\]
\end{example}
 \noindent Table 1 provides the expressions of $CRM_\alpha(X)$ for some well-known distributions.
\begin{table}
 \centering
 \caption{Expression of CRMHE for some distributions.}
 \label{tab: CRMHE_distributions}
\begin{tabular}{|c|c|c|c|}

\hline
&Distribution&$\bar{F}(x)$&$CRM_\alpha(X)$ \\

\hline
$(i)$&$U(0,a)$&$1-\frac{x}{a};0<x<a, a>0$&${\frac{1}{\alpha-1}}\left(\frac{a}{3-\alpha}-1 \right)$\\

\hline
$(ii)$&exp$(\lambda)$&$e^{-\lambda x};x\geq0, \lambda>0$&$\frac{1}{\alpha-1}\left(\frac{1}{\lambda(2-\alpha)}-1 \right)$\\
\hline
$(iii)$&Pareto$(k;a)$&${\left(\frac{k}{x}\right)}^a ;x\geq k, k>0, a>0$&$\frac{1}{\alpha-1}\left(\frac{k}{a(2-\alpha)-1}-1 \right)$\\
\hline
$(iv)$&Pareto $II$&${\left( 1+\frac{x}{a}\right)}^{-b};x\geq0,a>0,b>0$&$\frac{1}{\alpha-1}\left(\frac{a}{b(2-\alpha)-1}-1 \right)$\\
\hline
$(v)$&GPD&${\left( 1+\frac{ax}{b}\right)}^{-\left( 1+\frac{1}{a}\right)};x>0,a>-1,b>0$&$\frac{1}{\alpha-1}\left(\frac{b}{(2-\alpha)(1+a)-a}-1 \right)$\\
\hline
\end{tabular}
\end{table}
\section{Dynamic cumulative residual Mathai--Haubold entropy(DCRMHE)}
The study of duration is of interest in many branches of science, including reliability, survival analysis, actuarial science, economics, business, and several other fields. Therefore, in this section, we introduce the dynamic cumulative residual  Mathai--Haubold entropy and examine its properties.
\begin{definition}
Consider the lifetime of a component or system as $X$, and let it survive up to time $t$. In such cases, we consider the random variable $X_t=X-t|X>t$, which is time-dependent or dynamic with the survival function given by
\[
\bar{F}_t(x)=
\begin{cases}
    \frac{\bar{F}(x)}{\bar{F}(t)}, & \text{if $x>t$},\\
    1, & \text{$otherwise$}.
\end{cases}
\]
\end{definition}
\begin{definition}
 For a non-negative continuous random variable $X$ with survival function $\bar{F}(x)$, DCRMHE denoted by $CRM_\alpha(X;t)$ is defined as
\begin{align}\label{CRM t}
    CRM_\alpha(X;t)= \frac{1}{\alpha-1}\left( {\int_t^\infty} \left(\frac{\bar{F}(x)}{\bar{F}(t)}\right)^{2-\alpha} \,dx -1\right), \alpha \neq 1, 0<\alpha<2.
\end{align}
\end{definition}


\noindent Differentiating (\ref{CRM t}) with respect to $t$, we get the following relationship with hazard rate $h(t)$
\begin{align}\label{CRM't}
(\alpha-1)CRM'_\alpha(X;t)=(2-\alpha)h(t)\left[(\alpha-1)CRM_\alpha(X;t)+1\right]-1.
\end{align}
From the definition of DCRMHE, we can observe the following properties.
\begin{property}
For $t=0$, then $CRM_\alpha(X;t)=CRM_\alpha(X)$.
\end{property}
\begin{property}
 As $\alpha \to 1$,  $CRM_\alpha(X;t)$ becomes the dynamic cumulative residual entropy function given in (\ref{CREt}).
 \end{property}
 \begin{property}
 Consider the random variable $Y=aX+b$, where $a>0$ and $b\geq0$. Then we have
\[ CRM_\alpha(Y;t) = a\, CRM_\alpha\!\left(X;\frac{t-b}{a}\right)+\frac{a-1}{\alpha-1}; \quad t\geq b.
 \]
 \end{property}
 \begin{proof}
    We have 
 \[
CRM_\alpha(Y; t) = \frac{1}{\alpha-1} \left( \int_t^\infty \left( \frac{\bar{F}_Y(y)}{\bar{F}(t)} \right)^{2-\alpha} \, dy - 1 \right).
\]

 When $Y=aX+b$,
\[
CRM_\alpha(Y; t) = \frac{1}{\alpha-1} \left( a \int_{\frac{t-b}{a}}^\infty \left( \frac{\bar{F}(x)}{\bar{F}(t)} \right)^{2-\alpha} \, dx - 1 \right)
\]

\[
\qquad\qquad\qquad\qquad\qquad = \frac{1}{\alpha-1} \left( a \, CRM_\alpha\left(X; \frac{t-b}{a}\right)(\alpha-1) + a - 1 \right), \quad t \geq b.
\]

 Hence the proof.
 \end{proof}
 \begin{corollary}
 From the above property, it is clear that\\
$(i)$ If $b=0$, then $CRM_\alpha(aX;t)=aCRM_\alpha\left(X;\frac{t}
{a}\right)+\frac{a-1}{\alpha-1}$.\\
$(ii)$ If $a=1$, then $CRM_\alpha(X+b;t)=CRM_\alpha(X;t-b)$.
\end{corollary}
The following theorem shows that the dynamic cumulative residual  Mathai--Haubold entropy uniquely determines the distribution of $X$.

\begin{theorem}
 Consider a non-negative random variable $X$ with density function $f(x)$, survival function $\bar{F}(x)$ and hazard rate $h(x)$. Suppose that $CRM_\alpha(X;t)$ is increasing in $t$. Then $CRM_\alpha(X;t)$ uniquely determines the distribution function of $X$.
 \end{theorem}
 \begin{proof}
Suppose that $F(x)$ and $G(x)$ are distribution functions such that
 \begin{align}\label{CRM F=G}
CRM_\alpha(X;t)=CRM_\alpha(Y;t).
     \end{align}
This implies that
     \[
     {\int_t^\infty} \left(\frac{\bar{F}(x)}{\bar{F}(t)}\right)^{2-\alpha} \,dx={\int_t^\infty} \left(\frac{\bar{G}(x)}{\bar{G}(t)}\right)^{2-\alpha} \,dx.
     \]
     Differentiating the above equation with respect to $t$ and using the definition of hazard rates, we get
     \begin{align}\label{hazard}
h_{1}(t) \left[ (\alpha - 1) CRM_\alpha(X; t) + 1 \right]
= h_{2}(t) \left[ (\alpha - 1) CRM_\alpha(Y; t) + 1 \right].
\end{align}

     where $h_{1}(t)$ and $h_{2}(t)$ are the hazard rates corresponding to $f(x)$ and $g(x)$, respectively. To prove $\bar{F}(t)=\bar{G}(t)$, it is enough to show that $h_{1}(t)=h_{2}(t),$ for $t( \geq 0).$
     Using (\ref{CRM F=G}) in (\ref{hazard}), we get $ h_{1}(t)=h_{2}(t).$       This implies that $CRM_\alpha(X;t)$ uniquely determines the distribution of $X$.
 \end{proof}
\section{Characterization results}
In this section, we characterise some well known distributions based on the dynamic cumulative residual  Mathai--Haubold  entropy.
The following theorem shows that the dynamic cumulative residual Mathai--Haubold entropy is independent of $t$ if and only if $X$ follows an exponential distribution.
\begin{theorem}
  Let $X$ be a continuous non-negative random variable with distribution function $F(x)$. Then $CRM_\alpha(X; t)$ is independent of $t$ if and only if $X$ follows an exponential distribution.
  \end{theorem}\vspace{-0.2in}
  \begin{proof}
  Assume that $CRM_\alpha(X;t)=k$, where $k$ is a positive constant.\\
  Then
  \[
   CRM'_\alpha(X;t)=0.
  \]
   Using the relationship between $CRM_\alpha(X;t)$ and the hazard rate given in (\ref{CRM't}), we obtain
   \[
   (2-\alpha)\,h(t)\left((\alpha-1)k+1\right)-1=0.
   \]
   From the above expression, we observe that $h(t)$ is constant. Since a constant hazard rate characterizes the exponential distribution, we conclude that $X$ is exponentially distributed.
  \par Conversely, assume that X has an exponential distribution with parameter $\theta$, then
  \[
  CRM_\alpha(X;t) =\frac{1}{\alpha-1}\left( \frac{1}{\theta(2-\alpha)}-1\right),   \text{a constant.}
  \]
Hence, it is clear that DCRMHE is independent of $t$ if and only if $X$ has an exponential distribution.
\end{proof}
The following theorem provides a characterization of the generalized Pareto distribution based on the functional structure of the dynamic cumulative residual Mathai–Haubold entropy.
\begin{theorem}
    Consider a non-negative continuous random variable $X$ with survival function $\bar{F}(x)$, then the dynamic cumulative residual Mathai--Haubold entropy is a linear function of $t$ if and only if $X$ follows a generalized Pareto distribution (GPD).
    \end{theorem}\vspace{-0.2in}
    \begin{proof}
       Assume that $CRM_\alpha(X;t)=a+bt$, where $b \neq 0$. Then
\[
CRM'_\alpha(X;t)=b.
\]
Using the relationship given in (\ref{CRM't}), the above equation becomes
\begin{align}\label{gpd}
h(t)\left((\alpha-1)(a+bt)+1\right)
= \frac{(\alpha-1)b+1}{2-\alpha}.
\end{align}
Differentiating (\ref{gpd}) with respect to $t$, we obtain
\[
h(t)(\alpha-1)b
+ \left((\alpha-1)(a+bt)+1\right) h'(t)
= 0.
\]

    The above equation can be written as
      \[
     - \frac{d}{dt}\log h(t)=\frac{1}{k+t};k=\frac{a}{b}+\frac{1}{(\alpha-1)b}.
     \]
     Integrating both sides, we obtain
     \[
-\log h(t)=\log(k+t)+\log c.
     \]
That is
    \begin{equation*}
        h(t)=\frac{1}{ct+d}, \text{ where }  d=kc,
    \end{equation*}which is the hazard rate of GPD.

     Conversely, assume that $X$ follows GPD.
     Then
     \[\bar{F}(x)=\left(1+\frac{ax}{b}\right)^{-\left(1+\frac{1}{a}\right)}.\\
     \]
     Therefore, by direct computation, we obtain
     \[
      CRM_\alpha(X;t)=
     \frac{1}{\alpha-1}\left(\frac{b+at}{(a+1)(2-\alpha)-a}-1\right)
     \]
         \[ =k(b+at)-c,\]
         where \[k=\frac{1}{(\alpha-1)[(a+1)(2-\alpha)-a]} \text{ and }
             c=\frac{1}{\alpha-1}.\]
Thus, $CRM_\alpha(X;t)$ is linear in $t$. Hence, the theorem is proved.
\end{proof}
\noindent Next, we propose a characterization result based on the relation connecting the dynamic cumulative residual Mathai–Haubold entropy (DCRMHE) and the hazard rate $h(t)$.
  \begin{theorem}
    For a non-negative random variable $X$ with survival function $\bar{F}(x)$, hazard rate $h(t)$, and dynamic cumulative residual Mathai--Haubold entropy $CRM_\alpha(X;t)$, the relation
   \begin{align}\label{thm3.4}
   CRM_\alpha(X;t)=\frac{k}{h(t)}-c; c=\frac{1}{\alpha-1},
   \end{align}
   where $k$ is a positive constant, holds if and only if $X$ is a generalized Pareto random variable with survival function \[
   \bar{F}(x)=\left(1+\frac{ax}{b}\right)^{-\left(1+\frac{1}{a}\right)}; a>-1, b>0.
   \]
  \end{theorem}
  \begin{proof}
      Assume that (\ref{thm3.4}) holds. Then
      \[
      CRM'_\alpha(X;t)=-k(h(t))^{-2}h'(t).
      \]
      Using the relationship given in (\ref{CRM't})
      \begin{align}\label{proof3.4}
    -(\alpha-1)k(h(t))^{-2}h'(t)=(2-\alpha)\left((\alpha-1)k+h(t)(1-(\alpha-1)c)\right)-1.
      \end{align}
     Since $c=\frac{1}{\alpha-1}$, (\ref{proof3.4}) becomes
      \[
      \frac{d}{dt}\left(\frac{1}{h(t)}\right)=2-\alpha-\frac{1}{k(\alpha-1)}
      \]
      \[
        =2-\alpha-\frac{c}{k}.
      \]
      Integrating both sides, we get
      \[
      \frac{1}{h(t)}=\left(\frac{(2-\alpha)k-c}{k}\right)t+d_2.
      \]
    Therefore  \begin{align}\label{gpdh(t)}
     \ h(t)=\frac{1}{d_1t+d_2},
      \end{align}
      where $d_1=\frac{(2-\alpha)k-c}{k}$ and ${d_2}^{-1}=h(0)$.\\
       It should be noted that (\ref{gpdh(t)}) is the characteristic property of GPD as shown by Hall and
       Wellner (1981).\\
      Conversely, assume that $X$ follows GPD.
      From (\ref{CRM t}), we have
\[
 CRM_\alpha(X;t)=\frac{1}{\alpha-1}\left(\frac{b+at}{(a+1)(2-\alpha)-a}-1\right)
 \]
 \[
\qquad\qquad\qquad\qquad\qquad\qquad  =\frac{1}{\alpha-1}\left(\left(\frac{b+at}{a+1}\right)\left(\frac{a+1}{(a+1)(2-\alpha)-a}\right)-1\right)\]

 \[=\frac{k}{h(t)}-c,
            \]
            where
 \[
    k=\frac{a+1}{(\alpha-1)[(a+1)(2-\alpha)-a]} \text{ and } c=\frac{1}{\alpha-1}.\]
This completes the proof of the theorem.
\end{proof}
The following theorem provides a characterization result based on the relationship between the DCRMHE and the mean residual life function for the generalized Pareto distribution (GPD).
\begin{theorem}
    Let $X$ be a non-negative random variable with survival function $\bar{F}(x)$ and mean residual life function $m(t)$. Then the relationship
\begin{align}\label{thm3.5}
     CRM_\alpha(X;t)=km(t)-c, c=\frac{1}{\alpha-1},
\end{align}
holds for every $t>0$, if and only if X follows the GPD.
\end{theorem}
\begin{proof}
Suppose that (\ref{thm3.5}) holds and differentiating with respect to $t$, we get
        \[
  CRM'_\alpha(X;t)=km'(t).
        \]
      Using (\ref{thm3.5}) in (\ref{CRM't}), we get
\begin{align}\label{proof3.5}
      (\alpha-1)km'(t)=(2-\alpha)h(t)\left((\alpha-1)(km(t)-c)+1\right)-1.
      \end{align}
      Also, we have
      \[
      \frac{1+m'(t)}{m(t)}=h(t).
      \]
   Therefore, (\ref{proof3.5}) becomes
   \[
 (\alpha-1)km'(t)=(2-\alpha)\left((\alpha-1)k(1+m'(t))+\frac{1+m'(t)}{m(t)}(1-c(\alpha-1))\right)-1.
   \]
   Since $c=\frac{1}{\alpha-1}$, the above equation becomes
   \[
   m'(t)=c(2-\alpha)-\frac{c^2}{k}, \text{a constant.}
   \]
    This shows that $m(t)$ is linear in $t$. The linear mean residual life function is a characteristic property of GPD.\\
  Conversely, assume that $X$ follows GPD. By direct calculation, we obtain
  \[
   CRM_\alpha(X;t)=k(b+at)-c
   \]
   \[
\quad   =km(t)-c,
  \]where
  \[
  k=\frac{1}{(\alpha-1)((a+1)(2-\alpha)-a)}, c=\frac{1}{\alpha-1}.
  \]
Thus, the theorem is established.
\end{proof}
\section{New classes of lifetime distributions}
 In this section, two new classes of lifetime distributions in terms of dynamic cumulative residual Mathai-Haubold entropy, $CRM_\alpha(X;t)$ are given.\par

\begin{definition}
A random variable $X$ is said to have an increasing (decreasing) DCRMHE, denoted by  IDCRMHE (DDCRMHE), if $CRM_\alpha(X;t)$ increases (decreases) as $t$ increases for $t\geq 0$.
\end{definition}

\begin{definition}
A random variable $X$ is said to have an increasing (decreasing) failure rate IFR (DFR), if $h(t)$ increases (decreases) as $t$ increases,  $t\geq 0$.
\end{definition}
\begin{flushleft}
 The following theorem provides bound in terms of $h(t)$ for $CRM_\alpha(X;t)$.
\end{flushleft}
\begin{theorem}
    A random variable $X$ with distribution function $F(x)$ possesses an increasing (decreasing) DCRMHE if and only if for all $t>0$
    \[
   CRM_\alpha(X;t)\geq(\leq)\frac{1}{\alpha-1}\left(\frac{1}{(2-\alpha)h(t)}-1\right)
    \]
    or\[
    h(t)\geq(\leq)\frac{1}{(2-\alpha)[(\alpha-1)CRM_\alpha(X;t)+1]}.
    \]
    \end{theorem}
\begin{proof}
    The proof of the theorem follows directly from (\ref{CRM't}) and Definition 5.1.
\end{proof}
    The next theorem provides the hazard rate ordering that makes use of DCRMHE.
\begin{theorem}
    Consider two non-negative continuous random variables $X$ and $Y$ with survival functions $\bar{F}(t)$ and $\bar{G}(t)$ and hazard rate functions $h_1(t)$ and $h_2(t)$, respectively. Let $X\geq^{hr}Y$, that is, $h_1(t)\leq h_2(t)$ for all $t\geq 0$, then $ CRM_\alpha(X;t)\geq (\leq)  CRM_\alpha(Y;t)$ for all $1<\alpha<2$ $(0<\alpha<1)$.
    \end{theorem}
    \begin{proof}
        Assume that $h_1(t)\leq h_2(t)$.
        Then $\bar{F}_{X_t}(t)$$\geq$ $\bar{G}_{X_t}(t)$,
        which implies that
        \[
        \frac{\bar{F}(x)}{\bar{F}(t)} \geq \frac{\bar{G}(x)}{\bar{G}(t)}.
        \]
        Consequently
        \[
       \int_{t}^{\infty} {\left(\frac{\bar{F}(x)}{\bar{F}(t)}\right)}^{(2-\alpha)} dx\geq \int_{t}^{\infty}{\left( \frac{\bar{G}(x)}{\bar{G}(t)}\right)}^{(2-\alpha)}dx,
        \]
        and hence
    \[
        (\alpha-1)CRM_\alpha(X;t)+1\geq(\alpha-1)CRM_\alpha(Y;t)+1.
     \]
        For $1<\alpha<2$, this gives
        \[
        CRM_\alpha(X;t) \geq CRM_\alpha(Y;t),\]
      while for $0<\alpha<1$, we get
      \[CRM_\alpha(X;t) \leq CRM_\alpha(Y;t).\]
     Hence, the theorem.
    \end{proof}
\begin{theorem}
The uniform distribution on $(a,b),a<b$ can be distinguished by DDCRMHE for $1<\alpha<2$ and IDCRMHE for $0<\alpha<1$.
\end{theorem}
 \begin{proof}
        Let $X$ follows uniform distribution with parameters $a$ and $b$.\\  Then
        \[
        \bar{F}(x)=\frac{b-x}{b-a};a<x<b.
        \]
      \[
   CRM_\alpha(X;t)= \frac{1}{\alpha-1}\left(-\frac{(b-t)}{\alpha-3}-1\right).
      \]
      Differentiating with respect to $t$, we get
       \[
       CRM'_\alpha(X;t)=\frac{1}{(\alpha-1)(\alpha-3)}.
       \]
       It is clear that, when $0<\alpha<1$ \[CRM'_\alpha(X;t)>0\]
       and when $ 1<\alpha<2$ \[CRM'_\alpha(X;t)<0.
       \]
       Hence the theorem.
        \end{proof}
 \noindent The following theorem shows that the only distribution that is both IDCRMHE and DDCRMHE is exponential.
       \begin{theorem}
        Consider a random variable $X$ that holds both IDCRMHE and DDCRMHE, then $X$ follows an exponential distribution.
        \end{theorem}
        \begin{proof}
            If $X$ is DDCRMHE (IDCRMHE), then  \[
             CRM'_\alpha(X;t)\leq0(\geq0).
            \]
            From the above  inequalities, we get \[
            CRM_\alpha(X;t)=k, \\\ a \\\ constant.
            \]

            By Theorem 4.1, the constancy of $CRM_\alpha(X;t)$ means that $X$ follows an exponential distribution. This completes the proof.
        \end{proof}

       \begin{corollary}
              If the random variable $X$ has DDCRMHE (IDCRMHE), then
           \[
           \bar{F}(t)\geq(\leq)exp\left(-\int_{0}^{t}{\frac{1}{(2-\alpha)[(\alpha-1)CRM_\alpha(X;t)+1]}}dx\right).
           \]
  \end{corollary}
       \begin{proof}
           If $X$ possess DDCRMHE (IDCRMHE), then
           \[
           h(t)\leq (\geq)\frac{1}{(2-\alpha)[(\alpha-1)CRM_\alpha(X;t)+1]}.
           \]
           But
           \[
           \bar{F}(t)=exp\left(-\int_{0}^{t}h(u)du\right).
           \]
           Therefore
           \[
           \bar{F}(t)\geq (\leq) exp\left(-\int_{0}^{t}{\frac{1}{(2-\alpha)[(\alpha-1)CRM_\alpha(X;t)+1]}}dx\right).
           \]
       \end{proof}
\section{Non-parametric Kernel estimation}
Let $X_1,X_2,...,X_n$ be a random sample taken from a population with distribution function
$F$. Based on kernel density estimation, we construct non-parametric estimators for the proposed measures and let us assume that the kernel function $k(x)$ satisfies the following conditions:
\begin{itemize}
    \item[1.] $k(x) \geq 0$, for all $x$
    \item[2.] $\int {k(x)} dx =1$
    \item[3.] $k(.)$ is symmetric.
\end{itemize}
 The probability density function $f(x)$ at a point $x$ can be estimated using the kernel density estimator given by (Parzen (1962))
\begin{equation}\label{f x}
    f_n(x)=\frac{1}{nh}\sum_{j=1}^{n}k\left(\frac{x-X_j}{h}\right), \text{$h$ is the bandwidth.}
\end{equation}

Since our measure is defined based on survival functions, we choose a kernel based estimator for the survival function, which is given by
\begin{equation}\label{F x}
    \bar{F}(x)=\frac{1}{n}\sum_{j=1}^{n}\bar{K}\left(\frac{x-X_j}{h}\right),
\end{equation}
where $\bar{K}$ be the survival function of the kernel $k$ and $\bar{K}(t)=\int_{t}^{\infty}k(u)du$. \par

\sloppy
Non-parametrically, the CRMHE, $CRM_{\alpha}(X)$ can be estimated using a kernel-based approach, defined as follows
\begin{align}\label{CRMX hat}
    \widehat{CRM}_{\alpha}(X)=\frac{1}{\alpha-1}\left(\int_{0}^{\infty}\left(\frac{1}{n}\sum_{j=1}^{n}\bar{K}\left(\frac{x-X_j}{h}\right)\right)^{2-\alpha}dx-1\right).
\end{align}
Similarly, the estimator of DCRMHE, $CRM_{\alpha}(X;t)$ is as follows
\begin{align}\label{CRMt hat}
    \widehat{CRM}_{\alpha}(X;t)=\frac{1}{\alpha-1}\left(\int_{t}^{\infty}\left(\frac{\sum_{j=1}^{n}\bar{K}\left(\frac{x-X_j}{h}\right)}{\sum_{j=1}^{n}\bar{K}\left(\frac{t-X_j}{h}\right)}\right)^{2-\alpha}dx-1 \right).
\end{align}
Next, we examine the consistency of the proposed estimators. The kernel based estimator of the cumulative distribution function $F(x)$, established by Berg and Politis (2009), is consistent and is defined as
\begin{equation}\label{F hat}
    \widehat{F}_h(x)=\int_{-\infty}^{t}\hat{f}(t)dx=\frac{1}{n}\sum_{j=1}^{n}\tilde{K}\left(\frac{t-X_j}{h}\right),
\end{equation}

where $\tilde{K}(t)=\int_{0}^{t}k(u)du$.\par

 They have also given the expression for the variance of $\hat{F}_h(t)$ to establish consistency, which is given by
\begin{align}\label{var}
    Var(\widehat{F}_h(t))=\frac{F(t)(1-F(t))}{n}-\frac{2f(t)}{n}\left(\int u\tilde{K}(u)k(u)du\right)h+O\left(\frac{h}{n}\right).
\end{align}
Under certain assumptions, if $h\to 0$ and $nh \to \infty$ as $n \to \infty$, then $Var(\widehat{F}_h(t)) \to 0$ and hence the consistency of $\widehat{F}_h(t)$ is established. In order for the consistency of our proposed estimator to be proved, we need the following assumptions.\\
Let $\phi(t)$ denote the characteristic function of $X$:
\begin{itemize}
    \item[(A)] There is a $p >0$ such that $\int_{-\infty}^{\infty}|t|^{p}|\phi(t)|< \infty$.
    \item[(B)] There are positive constants $d$ and $D$ such that $|\phi(t)| \leq De^{-d|t|}$.
    \item[(C)] There is a positive constant $b$ such that $\phi(t)=0$ for all $|t| \geq b$.
\end{itemize}
    The consistency of the estimators has to be proved next. For this, first we prove the consistency of $\widehat{\bar{F}}(t)$. By direct calculation, it can be shown that
    \begin{equation*}
        \tilde{K}(t)=1-\int_t^\infty k(u)du=1-\bar{K}(t).
    \end{equation*}
    \begin{equation*}
    \text{So,} \\\ \widehat{F}(t)=\frac{1}{n}\sum_{j=1}^{n}\left(1-\bar{K}\left(\frac{t-X_j}{h}\right)\right)=1-\hat{\bar{F}}(t).
    \end{equation*}
    Therefore, (\ref{var}) becomes
    \begin{align*}
     Var(\widehat{\bar{F}}_h(t))=\frac{\bar{F}(t)(1-\bar{F}(t))}{n}-\frac{2f(t)}{n}\left(\int u(1-\bar{K}(u))k(u)du\right)h+O\left(\frac{h}{n}\right).
    \end{align*}
    \text{That is}
\begin{equation}\label{var final}Var(\widehat{\bar{F}}_h(t))=\frac{\bar{F}(t)(1-\bar{F}(t))}{n}+\frac{2f(t)}{n}\left(\int u\bar{K}(u)k(u)du\right)h+O\left(\frac{h}{n}\right).
    \end{equation}
    It is clear from (\ref{var final}) that $\widehat{\bar{F}}_h(t)$ is a consistent estimator of $\bar{F}(t)$. Hence from (\ref{CRMX hat}), we can see that $\widehat{CRM}_{\alpha}(X)$ is a consistent estimator of ${CRM}_{\alpha}(X)$ and from (\ref{CRMt hat}), $\widehat{CRM}_{\alpha}(X;t)$ is a consistent estimator of ${CRM}_{\alpha}(X;t)$.
\section{Simulation studies}
This section gives Monte Carlo simulation studies based on the estimators
$ \widehat{CRM}_{\alpha}(X)$ and $ \widehat{CRM}_{\alpha}(X;t)$. The aim of Monte Carlo simulation study is to evaluate the finite-sample performance of the proposed estimator by analyzing its bias and mean squared error under different sample sizes and distributions. Here, we use R software to perform the simulation and the experiment is repeated 10,000 times on various sample sizes $n=30, 50, 70, 90$. We consider two different lifetime random variables, Weibull and Uniform, from which the random variable $X$ is generated. The parameters are randomly selected and various sample sizes are used for different values of $\alpha$. To estimate the proposed measure, we use the kernel survival estimator. In order to select the bandwidth, Silverman's thumb rule is used, where the bandwidth is given by
$
h=1.06 \widehat\sigma {n}^ {-1/5},$ and $\widehat\sigma$ be the standard deviation of $n$ samples taken into account. With reference to equations (\ref{CRMX hat}) and (\ref{CRMt hat}), we calculate the estimates of CRMHE and DCRMHE, and in addition, the bias and MSE are also obtained.\par
Table \ref{CRMHE alpha=1.5} and Table \ref{CRMHE alpha=0.5} present the bias and MSE of the CRMHE estimator for the Uniform and Weibull distributions across various sample sizes for $\alpha=1.5$ and $\alpha=0.5$, respectively. From Table \ref{CRMHE alpha=1.5}, it is evident that the MSE values are lower for the Uniform distribution compared to the Weibull distribution. Similarly, from Table 3, it can be observed that the CRMHE estimator performs better for Uniform samples than for Weibull samples when $\alpha=0.5$.
For both values of $\alpha$, the bias and MSE decrease as the sample size $n$ increases. This show that the CRMHE estimator becomes more accurate and reliable with larger samples, irrespective of the underlying distribution.

\begin{table}[h!]
    \centering
    \caption{Bias and MSE of the CRMHE estimator for various distributions when $\alpha=1.5$}.
    \label{CRMHE alpha=1.5}
    \begin{tabular}{|c|cc|cc|}
\toprule
$n$ & \multicolumn{2}{c|}{$X \sim \text{Weibull}(5,1)$ } &  \multicolumn{2}{c|}{$X \sim \text{Uniform}(1.25,1.75)$}  \\

\cmidrule{1-5}

& Bias & MSE & Bias & MSE \\

\midrule
30 &0.0334& 0.0074& 0.0434& 0.0043\\

50& 0.0292& 0.0046 &0.0386& 0.0029\\

70& 0.0265& 0.0033&0.0349& 0.0022\\
90 &0.0240& 0.0026&0.0334& 0.0019 \\
\bottomrule
\end{tabular}
\end{table}

\begin{table}[h!]
\centering
\caption{Bias and MSE of the CRMHE estimator for various distributions when $\alpha=0.5$}.
    \label{CRMHE alpha=0.5}

\begin{tabular}{|c|cc|cc|}
\toprule
$n$ & \multicolumn{2}{c|}{$X \sim \text{Weibull}(5,1)$ }  &  \multicolumn{2}{c|}{$X \sim \text{Uniform}(0.5,1)$}  \\
\cmidrule{1-5}

& Bias & MSE & Bias & MSE \\
\midrule
 30& 0.0171& 0.0073& 0.0121& 0.0031\\

 50& 0.0146& 0.0043&0.0099& 0.0019\\
70& 0.0135& 0.0031& 0.0094& 0.0013 \\
 90& 0.0126& 0.0025&0.0080& 0.0011 \\
 \bottomrule
\end{tabular}
\end{table}

Next, we examine the performance of the DCRMHE estimator for different values of $t$ and $n$. Table \ref{DCRMHE alpha=1.5} gives the bias and MSE of the DCRMHE estimator when $X$ follows Weibull(5,3)  with $\alpha=1.5$. Similarly, Table \ref{DCRMHE alpha=0.5} presents the bias and MSE of the DCRMHE estimator when $X$ follows Weibull(5,1) with $\alpha=0.5$. It is evident from Tables 4 and 5 that the DCRMHE estimator shows lower bias and MSE for Weibull samples. From Table \ref{DCRMHE alpha=1.5} and Table \ref{DCRMHE alpha=0.5}, it is clear that the bias and MSE decrease as the sample size $n$ increases.
\begin{table}[h!]
\centering
\caption{Bias and MSE of the DCRMHE estimator across different $t$ and $n$ values for ${X \sim \text{Weibull}(5,3)}$ with $\alpha=1.5$}.
    \label{DCRMHE alpha=1.5}
\begin{tabular}{|c|c|cc|}
\toprule
$t$ & $n$ & \multicolumn{2}{c|}{$X \sim \text{Weibull}(5,3)$ } \\
\midrule
&  &Bias & MSE \\
\midrule
     &30 &0.1088 &0.0676\\
0.50      &50 &0.0938 &0.0417 \\
 &70 &0.0852 &0.0310 \\
     &90 &0.0790&0.0244 \\

     \midrule
     &30 &0.1134 &0.0685 \\
 0.75    &50 &0.0974 &0.0424 \\
 &70 &0.0882 &0.0315\\
     &90 &0.0817 &0.0248 \\

     \midrule
     &30 &0.1206 &0.0701 \\
     &50 &0.1031 &0.0435 \\
1.00 &70 &0.0932&0.0323 \\
     &90 &0.0860&0.0254\\

     \bottomrule
\end{tabular}
\end{table}

\begin{table}[h!]
\centering
\caption{Bias and MSE of the DCRMHE estimator across different $t$ and $n$ values for ${X \sim \text{Weibull}(5,1)}$ with $\alpha=0.5$}.
    \label{DCRMHE alpha=0.5}
\begin{tabular}{|c|c|cc|}
\toprule
$t$ & $n$ & \multicolumn{2}{c|}{$X \sim \text{Weibull}(5,1)$ }  \\
\midrule
&  &Bias & MSE \\
\midrule
       &30 &-0.0061 &0.0042   \\
0.50       &50 &-0.0041 &0.0026   \\
 &70 & -0.0034 &0.0019  \\
       &90 &-0.0029 & 0.0015  \\

       \midrule
       &30 &-0.0251 & 0.0030  \\
       0.75&50 &-0.0205 &0.0019   \\
        &70 &-0.0181&0.0014 \\
       &90 &-0.0164&0.0011   \\

       \midrule
       &30 &-0.0345&0.0027   \\
  1.00     &50 &-0.0291&0.0018   \\
 &70 &-0.0259 &0.0013   \\
       &90 &-0.0237 &0.0011   \\

        \bottomrule

\end{tabular}
\end{table}

\section{Data analysis}

In this section, two real-life datasets are used to analyze the proposed estimator given in (\ref{CRMt hat}). Specifically, first we consider the data used by Helu et al. (2020), which consists of  70 failure times of aircraft windshields measured in units of 1000 hours.
  Here, we estimated  the DCRMHE for the dataset with $\alpha=1.5$ and  assessed its efficiency by comparing the estimated and theoretical values of the measure for different values of $t$. 
   Table \ref{DCRMHEalpha=1.5} provides the bias and mean squared error (MSE) of DCRMHE using the kernel estimator defined in equation (\ref{CRMt hat}), based on 10,000 bootstrap samples of size 70 for $\alpha=1.5$. From Table \ref{DCRMHEalpha=1.5}, we can see that the proposed estimator and its theoretical values are closely matched, indicating that the accuracy of the estimation improves. This shows that for higher values of $t$, the level of uncertainty related to the failure time is reduced. In addition, we can see that the bias is moderately small compared to the true value. From fig.\ref{fig:graph 1}, we can see the plot of $CRM_{\alpha}(X;t)$ and $\widehat{CRM}_{\alpha}(X;t)$ against $t$ showing their decreasing trend for $\alpha=1.5$.

 \begin{table}[h!]
\centering
\caption{Bias and MSE of estimator for DCRMHE for different values of $t$ when $\alpha=1.5$}.
\label{DCRMHEalpha=1.5}
\begin{tabular}{c| c c c c}
\hline
$t$& $CRM_{\alpha}(X;t)$ & $\widehat{CRM}_{\alpha}(X;t)$ & Bias & MSE \\
\hline
0.9 & 1.9699 & 1.8861&-0.0913 &0.0368\\
1.0 & 1.7976&1.7156 &-0.0897& 0.0364\\

1.1 &  1.6349& 1.5505&-0.0922& 0.0366\\

1.2 & 1.4823&1.3914 &-0.0990 &0.0377\\

1.3 & 1.3403&1.2385 &-0.1101& 0.0397\\
\hline
\end{tabular}
\end{table}

\begin{figure}
    \centering
    \includegraphics[width=0.7\linewidth]{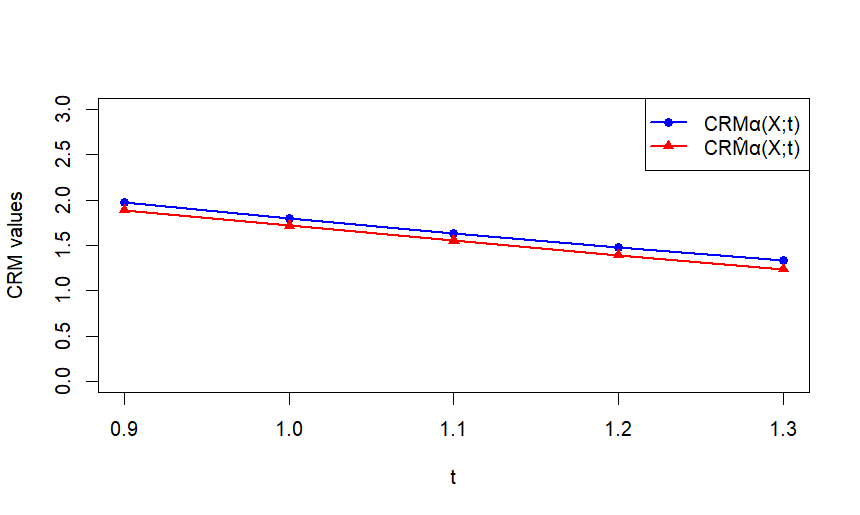}
    \caption{Comparison of $CRM_{\alpha}(X;t)$ and $\widehat{CRM}_{\alpha}(X;t)$ values when $\alpha=1.5$ for different $t$}.
    \label{fig:graph 1}
\end{figure}
Next, we consider the data used by Nair (1984) on the failure times (measured in millions of operations) of 40 randomly selected mechanical switches. Here we fit a Weibull model to the data and the fit is checked using Kolmogrov-Smirnov (KS) test. We fit the Weibulll distribution on the shape parameter $k= 3.85819$ and scale parameter $\lambda=2.3409$. The DCRMHE was estimated for the dataset when $\alpha=1.5$.  Table \ref{data 2} provides the bias and mean squared error (MSE) of DCRMHE using the kernel estimator defined in equation (\ref{CRMt hat}), based on 10,000 bootstrap samples of size 70. From Table \ref{data 2}, we can see that the proposed estimator and its theoretical values are closely matched, indicating that the accuracy of the estimation improves. From fig.\ref{fig:graph 2}, we can see the plot of $CRM_{\alpha}(X;t)$ and $\widehat{CRM}_{\alpha}(X;t)$ against $t$ showing their decreasing trend for $\alpha=1.5$.

\begin{table}[h!]
\centering
\caption{Bias and MSE of estimator for DCRMHE for different values of $t$ when $\alpha=1.5$}.
\label{data 2}
\begin{tabular}{c| c c c c}\hline
$t$& $CRM_{\alpha}(X;t)$ & $\widehat{CRM}_{\alpha}(X;t)$ & Bias & MSE \\
\hline
0.9 &  1.3142 & 1.4519&0.0907 &0.0758\\
1.0 & 1.1344&1.2741 &0.0921& 0.0767\\

1.1 &  0.9598& 1.1047&0.0965& 0.0781\\

1.2 & 0.7909&0.9442 &0.1041 &0.0799\\

1.3 & 0.6281&0.7926 &0.1142& 0.0825\\
\hline
\end{tabular}
\end{table}

\begin{figure}
    \centering
    \includegraphics[width=0.7\linewidth]{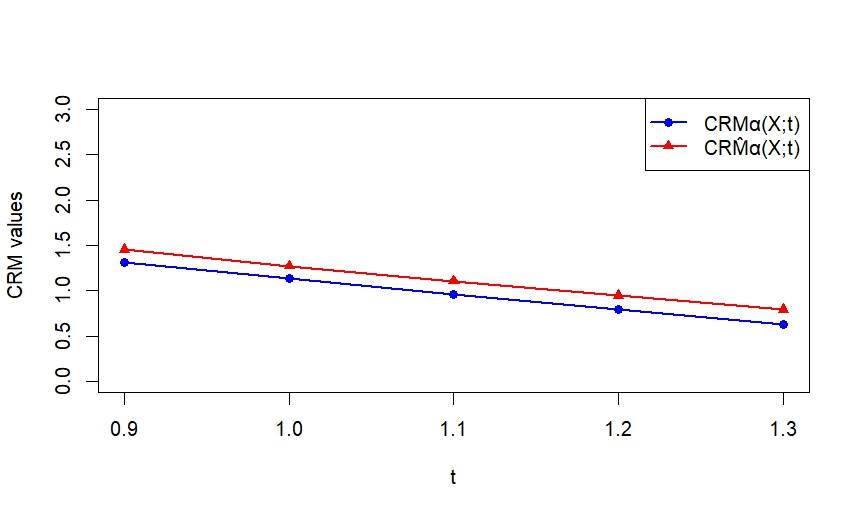}
    \caption{Comparison of $CRM_{\alpha}(X;t)$ and $\widehat{CRM}_{\alpha}(X;t)$ values when $\alpha=1.5$ for different $t$}
    \label{fig:graph 2}
\end{figure}

\section{Conclusion}
In this paper, we study the properties of the cumulative residual Mathai--Haubold entropy (CRMHE) and then propose a dynamic extension, DCRMHE, showing that it uniquely determines the distribution. Next, we examine the relationship linking the hazard rate and the mean residual life function with DCRMHE and characterize several lifetime distributions. We also investigate the non-parametric estimators of CRMHE and DCRMHE based on the kernel density estimation of the survival function, and their performance is assessed through a Monte Carlo simulation study. Finally, we used data on the failure times of 70 aircraft windshields and failure times of 40 randomly selected mechanical switches to demonstrate the relevance of the proposed DCRMHE estimator.

\end{document}